\documentclass[manuscript,screen, 10pt]{acmart}

\usepackage{algorithm}
\usepackage[noend]{algpseudocode}

\usepackage{mathrsfs}

\usepackage{subcaption}
\usepackage{graphicx}
\usepackage{lineno}

\AtBeginDocument{%
  \providecommand\BibTeX{{%
    \normalfont B\kern-0.5em{\scshape i\kern-0.25em b}\kern-0.8em\TeX}}}

\setcopyright{acmcopyright}
\copyrightyear{2018}
\acmYear{2018}
\acmDOI{XXXXXXX.XXXXXXX}

\acmConference[Conference acronym 'XX]{Make sure to enter the correct
  conference title from your rights confirmation emai}{June 03--05,
  2018}{Woodstock, NY}
\usepackage{color}
\definecolor{brandeisblue}{rgb}{0.0, 0.44, 1.0}

\newtheorem{remark}{Remark}
\def\VT{\mathit VT}
\def\FR{\mathit FR}
\def\CP{\mathit CP}
\def\A{\mathcal A}

\def\isdef{\mbox {$\ \stackrel{\rm def}{=} \ $}}

\def\LoR{\mathit LoR}
\def\GLoR{\mathit G(LoR)}
\def\TCB{\mathit TCB}
\def\RCB{\mathit RCB}
\def\ARA{\mathit ARA}
\def\SHA{\mathit SHA}
\def\pen{\varphi}
\begin{document}

\title{The Loop of the Rings: A Fully Decentralized Secure Cooperative System}


\author{Arash Vaezi}
\email{avaezi@sharif.edu}
\orcid{0000-0003-4798-0029}
\affiliation{
 \institution{Department of Computer Science and Engineering}
 \country{Sharif University of Technology}
}

\author{Amir Daneshgar}
\affiliation{
\institution{Department of Mathematical Sciences, Sharif University of Technology}
\email{daneshgar@sharif.edu}
}



\begin{abstract}
 We introduce $\LoR$, a secure, fully decentralized, and distributed cooperative system, where $\LoR$ stands for "the Loop of the Rings". 
 Distinct from conventional transaction-oriented systems, $\LoR$ prioritizes {\it cooperation} using its ring-based structure, making it possible to be used both as a cooperative workspace as well as a versatile platform for service provisioning, accommodating various roles such as freelancers, IoT management systems, and even managing 5G-related services.  Within this system, users have access to a secure and reliable environment, enabling them to offer a specific set of services to a potentially vast number of users. Our main contribution is to introduce the new structure along with its operating rules and principles, applying a combination of randomized procedures, in such a way that the whole system can be modeled fairly accurately in mathematically rigorous terms. This, in particular, is used to provide rigorous proof for the facts that $\LoR$ is both reliable and secure and that it may be efficiently implemented based on its typical communication complexity.   
\end{abstract}




\maketitle

\section{Introduction}
\label{sec:introduction}
Within the realm of collaboration systems, Douglas Engelbart~\cite{Douglas1962} pioneered the concept of collaborative computing.
A collaborative system can belong to a family of distributed applications. 
A collaborative system is usually designed to facilitate individuals working together towards a common task and achieving their shared objectives.
In such systems, there exists a critical need for close interaction among participants, involving the sharing of information and resources, exchanging work requests, and keeping track of each other's progress.


A variety of distributed cooperative systems has been proposed and studied in prior research~\cite{smartmicrogrids,multimedia,1045041}. In this context the underlying principle involves creating an interactive environment where users collaborate to achieve shared objectives. Such systems may encompass freelancing platforms, Internet of Things (IoT) networks, or even telecommunication networks. 

In 2008 Bocek et al.~\cite{Gametheoretical} designed a reputation-based incentive scheme for large-scale, fully decentralized peer-to-peer collaboration networks to encourage the participants to share their resources such as bandwidth and storage space and to edit and vote for documents that are shared to make history better and increase the reputation to get better services. 
In 2009 Bocek et al.~\cite{Voting} presented PeerVote, a decentralized voting mechanism in a peer-to-peer collaboration application. PeerVote is not blockchain-based and offers a strategy to maintain the quality of documents after each modification in the presence of malicious peers.
In 2019 M. Gandhi \emph{et al}.~\cite{GandhiFreelanceSystem} presented a practical implementation of the decentralized freelancing system based on blockchain named HireChain. Later, in 2021 I. Afrianto et al.~\cite{FreelanceSystem} introduced a prototype model of a freelance market system using blockchain technology based on smart contracts. Recently in 2022, K.~S.~Shilpa et al.~\cite{BlockchainFreelanceSystem} investigated how a freelancer marketplace can be implemented using Ethereum, blockchain, and smart contracts. 
These three works (HireChain and the next two ones) aimed to solve current issues of freelancing platforms, including unreliability, late payment, and delayed service.

In 2019 Bo Tang et al.~\cite{IoTPassport} published a blockchain-based trustful framework for collaborative IoT, named IoT Passport. It allows IoT platforms to construct arbitrary trustful relationships with one another, with precise criteria for intended partnerships enforced by a mix of smart contracts.
By the development of faster wireless technologies, such as 5G~\cite{shariatmadari2015machine}, IoT is anticipated to be quite popular.

Two other systems that fit into the category of cooperative systems are Coopedge and Algorand.
Cooperative edge computing is the key to full exploration and utilization of the power of edge computing. CoopEdge, introduced by Liang Yuan \emph{et al}.~\cite{Cooperative} is the first attempt to drive and support cooperative edge computing based on blockchain. CoopEdge has three main components, i.e., incentive mechanism, reputation system, and consensus mechanism. A potential security threat to the reputation system hired by CoopEdge is that a malicious edge server may try to fake the completion time while completing a peer-offloaded task. Such servers can also inject the reputation system with false ratings.

 The Algorand system, introduced by Yossi Gilad and team \cite{gilad2017algorand}, relies on a cryptographic selection method and the Byzantine agreement protocol. Encouraging user engagement and participation, which involves being online when selected and covering network costs, could be necessary. Incentives, possibly in the form of rewards, might be added to achieve this. Algorand also confronts the task of managing large data loads when new users retrieve existing blocks and their certificates. While other cryptocurrencies also deal with this issue, Algorand's high throughput may lead to scalability concerns. Notably, committee members' identities are exposed once they send messages, making gradual corruption of users a potential threat.

$\LoR$ is a secure collaborative, distributed system in which various groups of collaborators are formed, where each focuses on a specific job.
It is a convenient solution to 5G and IoT early challenges such as decentralization, network privacy, and security vulnerabilities~\cite{nguyen2020blockchain}.
It presents itself as a middleware, offering heightened reliability, security, and additional features to pre-existing cooperative systems. 
This system guarantees that when a user asserts the provision of a service, all other collaborators verify this claim, ideally within short timeframes. Furthermore, if a user invests money in collaboration, the system locks the required funds until the service is received and payment can be made to the workers.
Figure~\ref{fig:comparison} presents a summary of the distinctions among the aforementioned systems.

\setlength{\textfloatsep}{5pt}
\setlength{\intextsep}{10pt}
\begin{figure}[tp]
\centering
\includegraphics[width=1\textwidth]{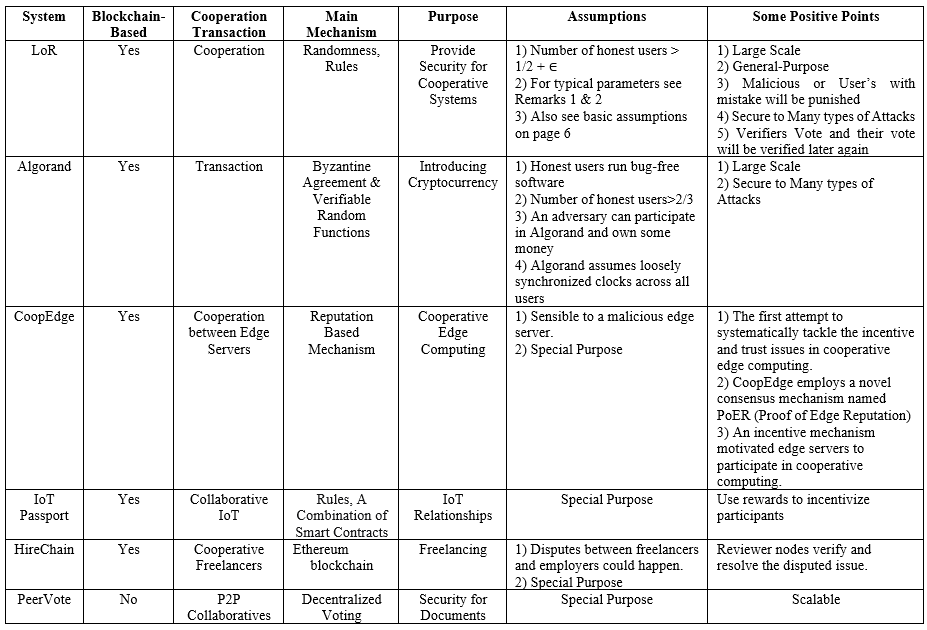}
\caption{A concise comparative overview.}
\label{fig:comparison}
\end{figure}

There are a lot of complex strategies used by the parties in the currently available systems to provide security~\cite{brown2018formal,carlsten2016instability,kiayias2016blockchain}. From another perspective, security within distributed systems introduces a concern due to the potential exposure of contents to other parties, including malicious ones. It involves concealing resources and data from unauthorized individuals, preventing unauthorized data alterations, and controlling access to data and computing services. Therefore, $\LoR$ must guard against attacks by malicious users to maintain its security. The system must be resistant to intervention from liars, ensuring reliable operation. See Appendix~\ref{sec:attacks} for details on the definition of security, reliability, and attacks.

Section~\ref{sec:summary} deals with a summary of the structural design of the $\LoR$ system. Section~\ref{sec:analysis} covers some mathematical analysis. 
Section~\ref{policies} states the policies. Section~\ref{sec:discussion} presents a small discussion.
Appendix~\ref{sec:properties} enumerates some of the properties and characteristics of $\LoR$. Appendix~\ref{sec:system} delves into the details of the components of the system. Appendix~\ref{sec:attacks} deals with the definition of attacks and security and the reasons why $\LoR$ is secure against those attacks.


\section{Design Overview}
\label{sec:summary}
The $\LoR$ system serves as a middleware, operating independently of the technical setup of a given platform 
designated for a specific application. When users utilize $\LoR$, they undergo iterative check-ups, while their funds are transferred through the $\LoR$ financial system. 
Therefore, the functioning of the $\LoR$ remains unaffected by factors such as communication protocols, whether the communication is asynchronous or synchronous, the collaborative infrastructure among parties, or even the type of network shared among the users of the application.
Hereafter, an implementation of the $\LoR$ system for a specific purpose is referred to as an \emph{instance} of $\LoR$.

The $\LoR$ system comprises the following components (the components are illustrated in Figure~\ref{fig.overal}):
\begin{enumerate}
    \item Administrator (Admin): The admin is responsible for selecting services and determining their prices. It is important to note that users are \emph{not allowed} to introduce new services themselves. Examples of services that administrators can define within their $\LoR$ instances include video streaming, storage provision, peer-reviewing, offering general services, and even hardware rentals such as CPU or GPU resources. Once the system is initiated, the admin's choices regarding services and prices are set and cannot be changed. 
   For more information on how the admin can select these services, please refer to Appendix Subsection \ref{servicerequirements}.
 
 \item  Traders: Individuals who already join the system are referred to as ``Traders'', as they can collaborate and exchange services. 
 When discussing a specific cooperation, traders may be identified as collaborators or coworkers (in what follows those who are not in $\LoR$ but intend to join an instance will be referred to as \emph{users}).

\item 
$\ARA$: For transactions, the $\LoR$ system employs an internal monetary system. It's important to emphasize that the intention behind this monetary system is NOT to introduce a new cryptocurrency. The currency unit in $\LoR$ is referred to as $\ARA$, which facilitates the exchange of services among collaborators within the platform. 

 \item  Coin: A {\it coin} is a structure that represents a particular service and its attributes. Every coin has a clear value based on $\ARA$. 

\item Cooperation Ring: This is a collection of coins owned by coworkers assembled randomly within a structure referred to as \emph{cooperation ring}s. 
In $\LoR$, in each cooperation ring, one trader pays for a service, and one or more other traders provide services. In what follows traders who pay are referred to as \emph{investor}s whereas traders who serve services are called \emph{worker}s. 
The size of a cooperation ring is determined by the type of services joined in that cooperation ring, while the type of a service is fixed by the administrator when the instance is initialized. 
In $\LoR$, the administrator can decide whether to select types of services having fixed cooperation ring sizes or allow services that have cooperation rings of arbitrary sizes.
Our analysis, as detailed in Section~\ref{sec:analysis}, covers both of these scenarios.
Note that, coworkers in a cooperation ring are unaware of their collaborators, as a random procedure selects them.
Also, a trader cannot have more than one coin in one specific cooperation ring. Traders are motivated to use one account though (see Rule~\ref{rule:oneaccount}).

 \item Fractal Ring: This refers to a collection of cooperation rings verified by a specific verification team. The size of this collection will be determined through a random procedure.

\item Verification Team: Corresponding to each fractal ring, there is a verification team. A group of randomly chosen traders is assigned the responsibility of meticulously inspecting each cooperation ring within a fractal ring to ensure trustworthy collaboration. This is the third and last random procedure of $\LoR$.
 \item Hierarchical Storage: 
A robust and reliable database is essential within the $\LoR$ system to monitor its status and facilitate information sharing. This involves establishing a hierarchical structure to safeguard critical data, including system status, account records, and essential system-related details. To achieve this, the system employs four distinct tables, each with a specialized role: the Coin Table, Cooperation Table, Ring-Control-Block ($\RCB$), and Traders-Control-Block ($\TCB$). Detailed explanations of these tables can be found in Section~\ref{sec:system}.
The $\TCB$ table takes center stage as the primary database, serving as a repository for comprehensive and enduring records encompassing the entire system. 
Blockchain technology can effectively implement this table, ensuring the secure storage of historical and essential data.
To ensure smooth operations, the $\TCB$ table is subject to mutual exclusion, a mechanism that prevents simultaneous write access by multiple verification teams, highlighting their pivotal role in system maintenance (see Subsection~\ref{subsec:TCBRCBCP} for more details).
 \end{enumerate}

The $\LoR$ system utilizes timestamps called a ``checkpoint'' to mark approximate time intervals within the system. These checkpoints establish the expected durations for the initiation and successful conclusion of collaborations.
Within the time span between consecutive checkpoints ($\CP$s), the system divides time into smaller units known as ``rounds''. The count of rounds between checkpoints can be dynamically adjusted according to system rules or specified by an administrator before the system instance begins (Remark~\ref{remark:typicalparameters} mentions the details of specifying round time). 
Each checkpoint triggers a process called the ``checkpoint-process'', while each round commences with a ``round-process''.

The $\LoR$ system employs an iterative random process that will be described below (for details see Section~\ref{sec:system}). These random processes and iterations through checkpoints ensure the proper functionality of the system even when the system is operating with unreliable traders, however, adhering to the system's rules and policies between checkpoints is essential for maintaining the reliability and security of $\LoR$ (refer to Section~\ref{policies} for detailed policy information).



When an administrator wishes to utilize $\LoR$ for a specific purpose, an instance of the system can be created based on the needs of the application's users.
In $\LoR$, each service is treated like a kind of currency, labeled as a \emph{coin}. For example, a trader who offers CPU services will have CPU-coins.

As mentioned before, $\LoR$ uses its own digital currency. Think of $\ARA$ like money in this system.
To participate in $\LoR$, a user needs to buy a certain amount of $\ARA$, at least as much as the cheapest coin available. Whether a trader is seeking a service or offering one, she/he needs to pay for her/his request. These requests are then converted into coin structures, and these coins randomly come together to form a cooperation ring. A cooperation ring involves at least two coins: an \emph{investor coin} and a \emph{worker coin} (see Subsection~\ref{subsec:cooperationrings} for the details of cooperation rings).
Traders can create multiple coins and take part in numerous cooperation rings, allowing those who invest more effort or resources to potentially earn more. 

Traders need to gather a random number of cooperation rings into a set and present it to the system as a fractal ring. Each fractal ring has a corresponding verification team which is assigned through a random selection process among traders.
 These teams are responsible for reviewing and validating cooperation rings along with their collaborators at regular rounds. 

A fractal ring must get submitted to the system for its cooperation rings to be able to start working and for the verification team to start their duties at checkpoints. 
This submission process serves first to validate all the cooperation rings contained within the fractal ring and second to store the information of the fractal ring in the $\TCB$ table.

Upon the successful completion of a verification process, the specific details pertaining to the newly submitted fractal ring are stored. Storing the fractal ring marks the conclusion of a successful submission process (refer to Figure~\ref{fig.overal}). For more detailed information regarding the submission process, refer to Rule~\ref{rule:VTpolicies}.
The submission and the previous steps occur during a checkpoint process, impacting the subsequent functionality of cooperation rings.
Once a fractal ring is submitted, all cooperation rings within it are designed to commence almost simultaneously and upon the following checkpoint (find the detailed process of constructing fractal rings in Subsection~\ref{subsec:fractalring}).

After passing a checkpoint, all the collaborators are expected to finish their jobs. Collaborators in a cooperation ring participating within a fractal ring can only receive their earnings after they receive confirmation from the verification team and upon the following checkpoint. Managing payments is one of the responsibilities of the verification teams. For detailed rules about payments, see Section~\ref{policies}.


Let's freeze a moment within the system. At this specific time, various fractal rings and other structures come into play. Construction of fractal rings is a responsibility of traders, involving the creation of coins, cooperation rings, and fractal rings. Verification teams oversee the validation of these structures. Throughout these processes, utilization of irreversible and non-predictable random methods is a fundamental characteristic.



In any instance of the $\LoR$ the following are assumed.  
\paragraph{\bf Basic assumptions:}
\begin{enumerate}
    \item The value $0.5+\epsilon < 1-\alpha \leq 1$ stands for the percentage of trustworthy  traders ($\epsilon$ is a small positive constant).
    \item All (or the majority of) the traders are able to provide every service defined in the system. A trader can choose to invest in any one of these services too. 
    \item Services need to be quantifiable in terms of both cost and operational aspects relative to time (Subsection~\ref{servicerequirements} in Appendix includes more details).

   \item Traders cannot define any different type of services to the system or change the prices of the services.
   \item $\TCB$ is distributed among the traders. Every memory of a verification team has a replica of the information required for verifying the corresponding fractal ring. The table $\TCB$ is implemented based on blockchain.
\end{enumerate}


\setlength{\textfloatsep}{5pt}
\setlength{\intextsep}{10pt}
\begin{figure}
\centering
\includegraphics[width=0.9\textwidth]{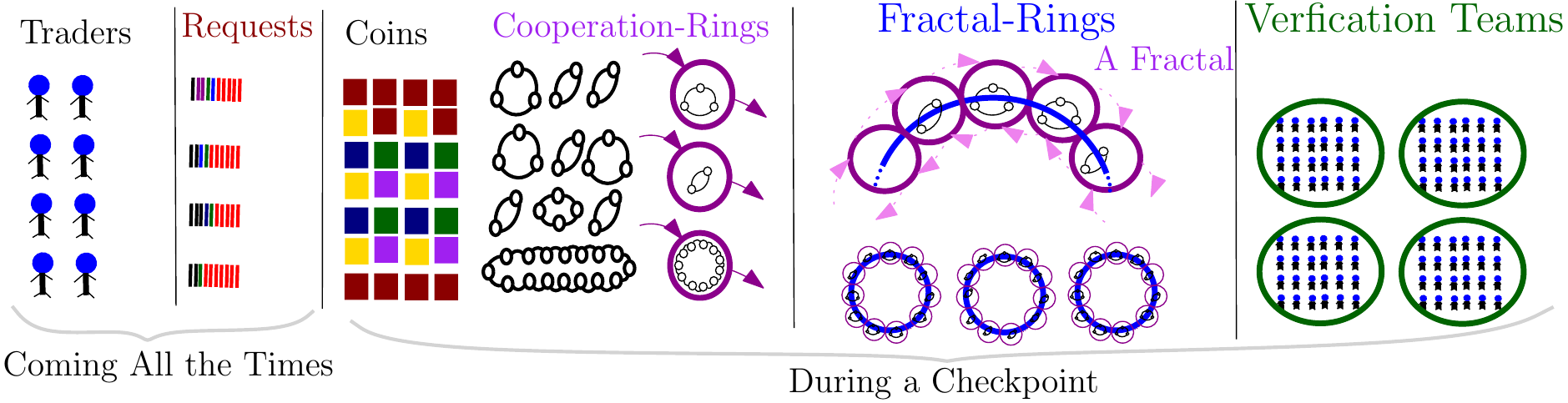}
\caption{A framed moment of the system.}
\label{fig.overal}
\end{figure}



\paragraph{\bf Characteristics}
Due to Theorem~\ref{thm:Verifsound} the verification teams are highly reliable.
Moreover, the decentralized random structure of the procedures, as well as an intentional lock mechanism used in $\LoR$, makes this platform highly secure against malicious traders.
The lock mechanism locks the coins in the memory of collaborators working in a cooperation ring so that the coworkers cannot modify or access the coins (see Subsection~\ref{subsubsec:lock} in Appendix for more information on the details of the lock mechanism).
Using this mechanism, the rules and the randomized structure of $\LoR$ incentivize honesty.
 (For more on the security of $\LoR$ against different attacks see Appendix~\ref{sec:attacks}, and Property~\ref{prop:security} in Appendix~\ref{subsec:characteristics}.)

Some other featues and characteristics of $\LoR$ include the following:
\begin{itemize}
\item Unlike previous systems like Algorand, which primarily focus on transactions, cooperation in the $\LoR$ system holds broader significance.

\item In $\LoR$, slow users are not disadvantaged.

\item Collusion and inflation are not possible within the system.

\item Traders have a strong incentive to collaborate across different services.

\item Every trader is obligated to actively participate and contribute to earn income.
\end{itemize}

For more detailed information on the characteristics of $\LoR$, please see Subsection~\ref{subsec:characteristics}.

\section{ Analysis of Reliability and Communication Complexity}
\label{sec:analysis}
This section delves into the analysis of various aspects of the $\LoR$ system. The highlights of this section are providing a fairly accurate abstract model of the $\LoR$ that paves the way to prove some important properties of this collaborative framework.
throughout this section, we let $\kappa$, which is an odd number, stand for the maximum size of a verification team, and $\ell$ stand for the maximum number of cooperation teams that a trader may join. Also, we assume that $N$ is the whole number of traders in a fixed checkpoint through which we are going to analyze the system. Within this setting $\alpha$ stands for the fraction of wrongdoers. It is instructive to note that by a {\it wrongdoer} we mean a trader that just does not perform his responsibilities as he is supposed to do based on his commitments and the principles of $\LoR$. Clearly, an attacker who intentionally attacks the system is definitely a wrongdoer but the reverse implication does not necessarily hold, meaning that the set of wrongdoers cover a very general subset of traders that may cause any kind of dissatisfaction when $\LoR$ is running.

Our first claim is related to the soundness of the system as far as the verification teams are concerned.

\begin{theorem}
\label{thm:Verifsound}
The probability that a verification team casts a wrong vote is negligible.
\end{theorem}
\begin{proof}
Since verification teams operate by majority vote the probability of casting a wrong vote is equal to 
$$
\mathbb{P}(Bin(\kappa,1-\alpha) \leq \lfloor\kappa/2\rfloor)=\sum_{_{i=0}}^{\lfloor\kappa/2\rfloor} {\kappa \choose i} (1-\alpha)^i \alpha^{\kappa-i},
$$ 
and by Berry-Esseen theorem we have,
$$
\left| \mathbb{P} \left(Bin(\kappa,1-\alpha) \leq \lfloor\kappa/2\rfloor \right) - \Phi\left( \left( \frac{\lfloor\kappa/2\rfloor}{\kappa} - (1-\alpha) \right) \frac{\sqrt{\kappa}}{\sqrt{\alpha(1 - \alpha)}}\right) \right| \leq C_{_{0}} \frac{(\alpha^2 + (1-\alpha)^2)}{\sqrt{\alpha(1 - \alpha)} \sqrt{\kappa}},
$$
in which $\Phi$ is the cumulative distribution function of the standard normal distribution, defined as
$$
\Phi(x) \isdef \frac{1}{\sqrt{2 \pi}}\int_{-\infty}^{x} e^{\frac{-z^2}{2}} \ dz
$$
and $C_{_{0}} \leq 0.41$ \cite{zolotukhin_nagaev_chebotarev_2018}. Now, if $\kappa$ is large enough then the probability of casting a wrong vote may be approximated by 
$$\Phi\left( \left( \alpha - \frac{1}{2} \right) \frac{\sqrt{\kappa}}{\sqrt{\alpha(1 - \alpha)}}\right)$$
showing that if $\alpha \ll \frac{1}{2}$ this probability is negligible.

\end{proof}

\subsection{Communication complexity and memory usage in $\LoR$}
\label{subsec:commmunicationcomplexity}
Although $\LoR$ is a fully decentralized system and its evolution is based on procedures which are light algorithms, one may be curious about the number of messages and the space that each trader communicates/needs while trading. 
To concentrate on the worst case, we consider an instance of $\LoR$ in which there are services of any possible size for the cooperation rings, and subject to this hypothesis we provide a theoretical model that will not only help to prove the boundedness of the communication complexity of the system (actually in an extreme worst case regime) but also will be effectively used later to analyze the satisfiability of traders as far as their operations subject to $\LoR$ principles are justified. 

\begin{theorem}
\label{thm:constantdgeree}
 The expected number of communication channels for a typical trader in $\LoR$ is asymptotically upper bounded by $\mathcal{O}(\kappa) + \mathcal{O}(\ell^{4})$ in the extreme worst case.
\end{theorem}

\begin{proof}

For this, let us first provide a graph theoretic model of the problem. In this regard, consider the dynamic multigraph of traders in which nodes correspond to traders and edges represent communication channels. To be able to go through a fair average-case analysis, we freeze this dynamic multigraph at a framed moment of the system between two checkpoints and note that the degree of a node in this static multigraph, hereafter referred to as  $\GLoR$, is fixed. Keeping in mind that this static multigraph is actually a random multigraph, we intend to provide an upper bound for the largest degree as a measure proportional to the size of data that a typical trader needs to store and the average number of the messages a trader needs to communicate between two checkpoints.

At first, we try to provide an approximation model in which random events are decomposed into almost independent simpler random processes to make a rigorous analysis possible.
Note that the two processes of formation of fractal rings and assigning verification teams naturally may be assumed to be independent. On the other hand, based on reality, it is logical to assume that traders do not discuss trades simultaneously, or even if they do, all different negotiations of a single trader for different cooperation teams are not related to each other, and consequently, are almost independent. This justifies that one may approximately model the random process of cooperation and fractal rings evolution by an amalgam of independent random permutations composed of cycles corresponding to cooperation rings  (see Figure~\ref{fig:momentslices} for a typical configuration of $\GLoR$).

\setlength{\textfloatsep}{5pt}
\setlength{\intextsep}{5pt}
\begin{figure}
     \begin{subfigure}[b]{0.35\textwidth}
         \includegraphics[width=\textwidth]{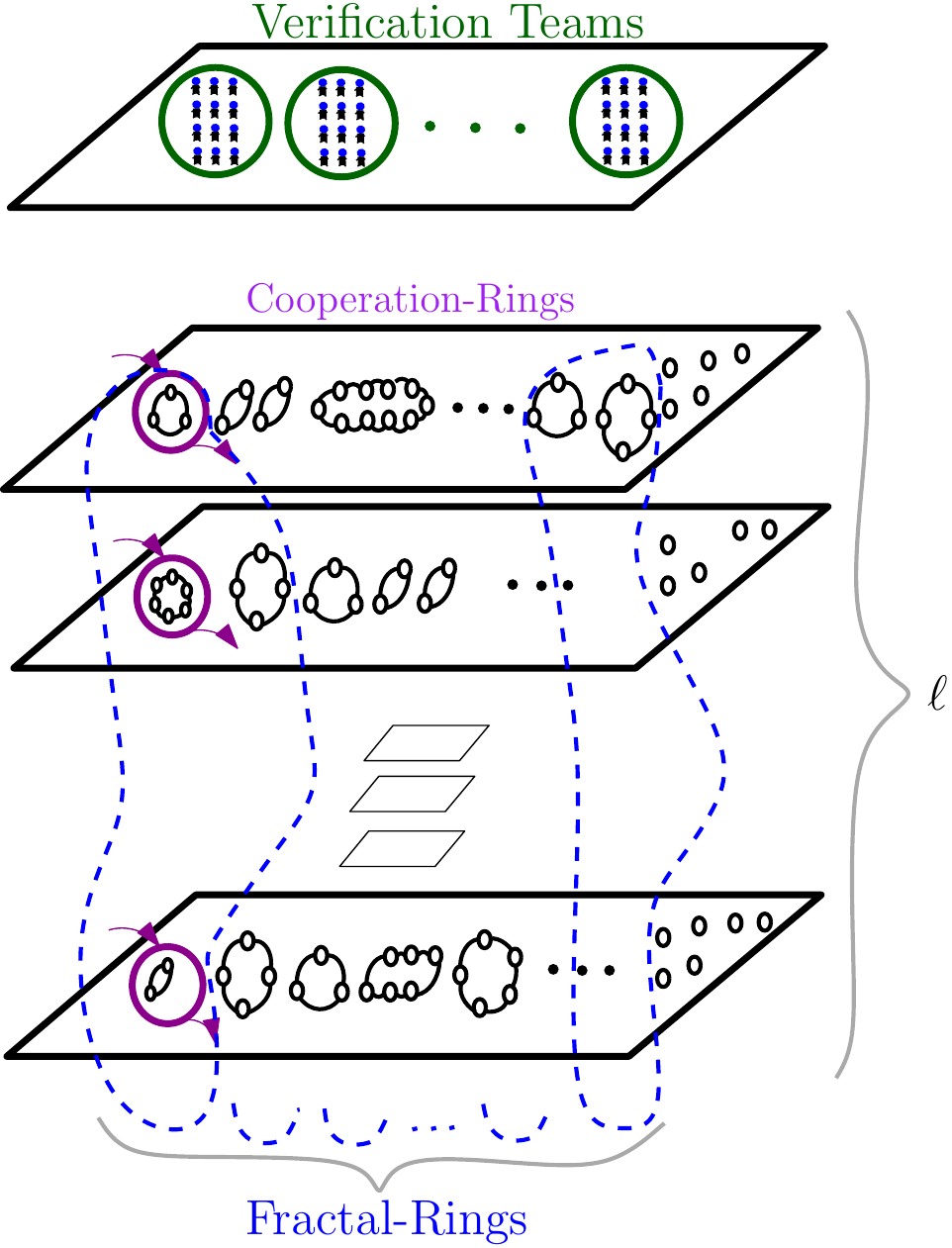}
         \caption{}
        \label{fig:momentslices}
     \end{subfigure}
     \hfill
     \begin{subfigure}[b]{0.4\textwidth}
         \includegraphics[width=\textwidth]{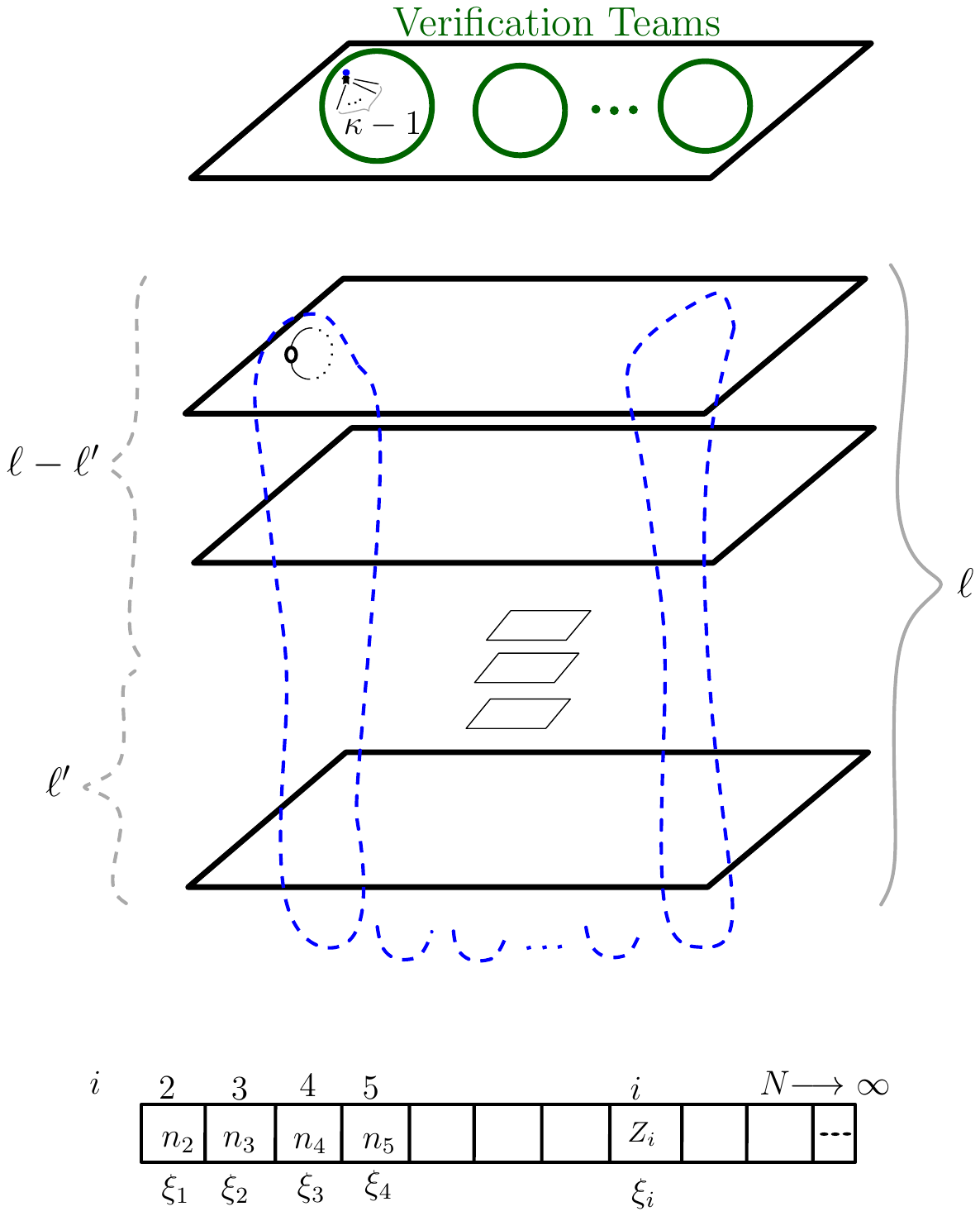}
         \caption{}
         \label{fig:expecteddegree}
     \end{subfigure}
     \hfill
        \caption{(a) A single frame of the system between two checkpoints, split into $\ell+1$ random slices. (b) \small{Let $n_{i}$ be the number of cooperation rings of size $i$ that obeys a Poisson distribution $Z_{i}$ having the mean parameter $\lambda_{i} \isdef \ell i^{-1}$, where a trader obeys a binomial distribution with the success probability $\xi_{i}$ to choose from each $n_{i}$ cooperation rings of size $i$ to form a fractal ring.}}
\end{figure}


Keeping these assumptions in mind and letting $\ell$ be the maximum number of cooperation rings that a trader may share, one may think of $\GLoR$ as an amalgam of $\ell+1$ slices, first of which containing the verification teams and $\ell$ next slices providing a representation of cooperation rings decomposed into $\ell$ independent random permutations. Within this setup, if a trader cooperates in $\ell' \leq \ell$ cooperation rings then the trader is represented as a fixed point of the rest of the corresponding $\ell-\ell'$ permutations.
Hereafter, we naturally assume that all traders in a verification team are connected to each other and that the size of all verification teams is equal to a constant (maximum) number $\kappa$. The slice containing verification teams may be thought of as an average case approximation of the reality, in which we have assumed that each trader on average cooperates in one verification team and each verification team verifies a unique fractal ring consisting of a block of cooperation rings in lower $\ell$ remaining slices\footnote{Refinements of our simplifying assumptions to more realistic scenarios, e.g. as weighted-multigraph models or uniform partitions of cooperation rings to fractal rings or even more complex verification team assignments, give rise to more involved models whose analysis are far beyond our limitations for what is discussed in this article and may appear elsewhere.}. In this regard note that, our decomposition is also justified by the randomness of the model in correlation with a result of S.~Janson \cite{JansonBook} indicating that the distribution of cycles of a constant length in an amalgam of random permutations still obeys a Poisson distribution.

Considering the setup of our simplified model, note that the degree of a typical node (say $v$), corresponding to a trader, is equal to $(\kappa-1)+2(\ell-s)+\mu$, 
 that consists of three parts. First note that, there are $\kappa-1$ edges within the verification team in which nodes are all connected to each other, and moreover, there are $2\ell'=2(\ell-s)$ edges within the $\ell'$ slices corresponding to cooperation rings, where $s$ is the number of slices in which the node is fixed by the corresponding permutation (meaning that the corresponding trader does not contribute to any cooperation ring in these $\ell'$ slices). On the other hand, there are also $\mu$ edges corresponding to connections between the node $v$ and the fractal ring verified by it. It is instructive to note that, in this setup, the probability of a node being fixed by a random permutation is equal to $1/N$ when $N$ is the total number of nodes, showing that the first component of the degree which is equal to $\kappa-1+2(\ell-s)$ tends to $\kappa+2\ell-1$ when $N$ tends to infinity. Hence, for a sufficiently large number of nodes (i.e. traders), this part may be assumed to be a constant. In the sequel, we are looking for an upper bound on $\mu$. 


To estimate the expectation of the random variable $\mu$, forming the random component of the degree, we apply the theory of logarithmic combinatorial structures (e.g. see \cite{TavareBook}). In this context, let $C_{i}(n)$ stand for 
the number of components of size $i$ of a specific combinatorial structure of size $n$.
 Taking uniform samples, one is interested to study the distribution of the $n$-tuple 
$$(C_{1}(n),C_{2}(n),\ldots,C_{n}(n)),$$
when $n$ is sufficiently large. The general idea in such an approach is to show that there exists a distribution $Z_{i}$ such that for the random variable 
$$T_{n} \isdef Z_{1}+2Z_{2}+\cdots+nZ_{n}$$
for which $Z_{i}$'s are i.i.d, one has the conditioning identity
$$(C_{1}(n),C_{2}(n),\ldots,C_{n}(n)) =_{d} ((Z_{1},Z_{2},\cdots,Z_{n})|T_{n}=n).$$
Within this setup, usually one tries to approximate a functional $\eta$ of the structural counts as 
$$\eta((C_{1}(n),C_{2}(n),\ldots,C_{n}(n))$$
by the result of applying the functional to the i.i.d process as
$$\eta((Z_{1},Z_{2},\cdots,Z_{n})).$$
Such an approach is usually quite hard to be analyzed rigorously and is extensively studied for some simple functionals on structural counts (see \cite{TavareBook} for the details).
It is also interesting to note that for the logarithmic combinatorial structures the limit $i \mathbb{E}(Z_{i})$ exists and can be computed in some specific cases.
 
In our case, we are actually facing a combination of such random processes in which the random permutation process falls into the class of random assemblies that may be approximated by i.i.d Poisson distributions, on the one hand, and $\eta$ itself is a random functional of choosing a uniform fractal ring, as a collection of cycles of random permutations, i.e. a random combinatorial process of selection type which may be approximated by binomial distributions.

Hence, first note that for the case of permutations we know that the distribution of cycles of size $i$ is a Poisson distribution with parameter $i^{-1}$, and consequently, for $\ell$ such independent random permutations, the distribution of cycles of size $i$ will be a Poisson distribution with the mean parameter $\lambda_{i} \isdef \ell i^{-1}$. Therefore, our setup reduces to choosing cycles of size $i$ from  $N$  Poisson distributions with parameter $\lambda_{i}$ and then picking a pack of cycles, as a fractal ring, uniformly at random from such a collection of cycles with $(C_{1}(n),C_{2}(n),\ldots,C_{n}(n)$ denoting the cycle counts for such a selection that may be approximated by binomial distributions (see \cite{TavareBook}). Clearly, the random variable of interest is
$$\mu \isdef C_{1}(n)+2C_{2}(n)+\ldots+nC_{n}(n).$$
In what follows we assume an asymptotic regime for sufficiently large $N$ as the number of traders in which case it is known that the distribution of cycle counts of a random permutation is asymptotically equivalent to that of the random variable 
$(Z_{1},Z_{2},\cdots),$
where in our case $Z_{i}$'s are all independent Poisson random variables with the mean parameter $\lambda_{i}$. Hence, assuming that the choice of cycles of length $i$ 
obeys a binomial random variable of probability $\xi_{i}$, one may compute the expectation of $\mu$ using the probability generating function of the Poisson distribution as follows,
$$\mathbb{E}(\mu)=\displaystyle{\sum_{i > 1}} i\ \xi_i \left(\displaystyle{\sum_{k > 1}} k\ \mathbb{P}(Z_i=k)\right).$$

In summary, at a fixed moment of the system, to construct a fractal ring, a trader chooses from cooperation rings of different lengths, say $i > 1$, that are generated by a Poisson random variable $Z_{i}$ having the parameter $\lambda_{i} \isdef \ell i^{-1}$, obeying a binomial distribution with the success probability $\xi_i$ (see Figure~\ref{fig:expecteddegree}). 

Now, to compute $\mathbb{E}(\mu)$, first, note that one may use the probability generating function of the Poisson distribution 
$$\displaystyle{ \sum_{k \geq 0}^{}P[z_i = k]x^k} = e^{-\lambda_i(1-x)}$$
to compute  $\displaystyle{\sum_{k > 1}} k\ \mathbb{P}(Z_i=k)$ as follows,
$$\frac{\partial{}}{\partial{x}}(\displaystyle{\sum_{k \geq 0}} e^{-\lambda_i} \cdot \frac{\lambda_i}{k!} x^k)\bigg|_{x=1} = \displaystyle{\sum_{k \geq 1} k (e^{-\lambda_i} \frac{\lambda^{k}_{i}}{k!})} = \lambda_i e^{-\lambda_i} + \displaystyle{\sum_{k > 1}} k (e^{\lambda_i} \cdot \frac{\lambda^{k}_{i}}{k!})
= \frac{\partial{}}{\partial{x}}(e^{-\lambda_i(1-x)})\bigg|_{x=1} = \lambda_i e^{\lambda_i},$$
implying that
$$ \displaystyle{\sum_{k > 1}} k (e^{-\lambda_i} \cdot \frac{\lambda^{k}_{i}}{k!}) = \lambda_i e^{\lambda_{i}} - \lambda_i e^{-\lambda_i} = \lambda_i(e^{\lambda_i} - e^{-\lambda_i}) = 2\lambda_i \sinh{(\lambda_i)}.$$
Consequently, 
\begin{equation}
\label{eq:main}
\mathbb{E}(\mu)=\displaystyle{\sum_{i > 1}} i\ \xi_i \left(\displaystyle{\sum_{k > 1}} k\ \mathbb{P}(Z_i=k)\right)=
 2\ell \displaystyle{\sum_{i > 1}} \xi_i \sinh(\ell/i).
\end{equation}

By noting that in reality, to form a fractal ring, the probability of choosing a ring of length $i$, i.e. $\xi_i$, fast decreases with respect to $i$,  one may assume that $\xi_i$ tends to zero much faster than $\lambda_i$ and by the inequality $\sinh(x) \leq x+x^3/5$ for $0 < x <1$, one may deduce that 
$$\mathbb{E}(\mu) =  2\ell \displaystyle{\sum_{i > 1}} \xi_i \sinh(\ell/i) \leq 2 \ell \left(\displaystyle{\sum_{1 < i \leq \ell}} \xi_i \sinh(\ell/i)+\displaystyle{\sum_{i > \ell}} \xi_i \sinh(\ell/i) \right)$$
$$ \leq 2 \ell \left(\displaystyle{\sum_{1 < i \leq \ell}} \xi_i \sinh(\ell/i)  +  
\displaystyle{\sum_{i > \ell}} \xi_i \frac{\ell}{i}
+ \displaystyle{\sum_{i > \ell}} \xi_i \frac{\ell^3}{5i^3} 
\right).$$
Now, letting $\xi_i$ be proportional to $(\sinh(\ell/i))^{-1}$ for $i \leq \ell$ and defining $\xi_i$ to be proportional to $i^{-1}$ for $i > \ell$ and recalling the definition of the Riemann zeta function as
$$\zeta(k) \isdef \displaystyle{\sum_{i \geq 1}} \frac{1}{i^k},$$
we have
$$\mathbb{E}(\mu) \leq 2 \ell (O(1)+O(1)\zeta(2)\ell+O(1)\zeta(4)\ell^3) \leq O(\ell^4),$$
showing that the worst case maximum expected degree in $\GLoR$ is upper bounded by $O(\kappa)+O(\ell^4)$.
\end{proof}

\subsection{Traders' satisfaction and reliability of $\LoR$}
\label{subsec:badusers}
 Note that in $\LoR$ there's no requirement to differentiate between a trader intentionally causing harm or a user whose work doesn't meet expectations unintentionally. The system functions effectively due to the absence of recurring mistakes among trustworthy individuals. If someone's actions fall short of expectations, they face temporary consequences. Consistently, unreliable individuals experience progressive penalties, while others gain advantages from the system. Users strive to excel in their work and increase their earnings. 
 As a result, a reliable trader is someone whose actions mostly meet the satisfaction of their collaborators. In what follows we prove that this is actually the case when the prenalty imposed of cooperation rings at each round is set to be proportional to the square of the inverse of the size of the cooperation ring or if there is an upperbound on the size of the cooperation rings whose square is sufficiently smaller than the whole number of traders. In what follows, $R$ is the number of rounds between two checkpoints, $\pen$ is the forfeit or penalty factor set to be imposed on the whole cooperation ring in case of dissatisfaction majority vote of the corresponding verification team (applied to the value of the trader's coin), where we always assume that the fraction of wrongdoers to the total number of traders in the whole instance of the system being analyzed is fixed and equal to $0 \leq \alpha < 1$ (note that this is essentialy a worst case compared to the case when the rate of wrongdoers decrease by an increase in the number of whole traders). 

Consider an individual who intends to engage in wrongdoing or provide false information about a service to exploit cooperation rings. This wrongdoer must participate in a cooperation ring or involve himself in a service. In either scenario, based on the specific service and cooperation, the wrongdoer needs to acquire a coin, which is secured by the lock strategy of $\LoR$. If a situation arises where at least one collaborator is not satisfied, and a trader declares dissatisfaction, the verification team halts the cooperation, and as a consequence, all collaborators face penalties, leading to a reduction in their coin-based payment at the next checkpoint related to that activity.


To facilitate our analysis, let's recall some key assumptions. There are two types of traders, namely, investors and workers, both of whom make payments in exchange for services. An investor, upon making an investment, anticipates receiving a specific service, while a worker expects to be compensated for her/his efforts. In this context, we use the term "satisfaction" instead of "profit" since satisfaction represents the degree to which a trader's expectations are met regardless of trader's intentions. 
It's already assumed that within each checkpoint of the system, a fraction of $1-\alpha$ percent of the traders are reliable.

Now, consider the structure of each checkpoint, which consists of a total of $R$ rounds. In the event of an unsatisfactory round experienced by any of the collaborators within a cooperation ring, the system imposes a penalty on all coworkers in that ring. Each coworker owns one coin in the corresponding cooperation ring and the penalty $\pen$ is set to be  proportional to $\frac{1}{R}$. 


In the sequel we call a cooperation ring, $CR$, to be 
a \emph{bad} corporation rng if it contains at least one \emph{bad} trader, i.e. a wrongdoer.
Also, let us call a verification team, $\VT$, to be a \emph{bad} verification team if at least half of its members are bad traders. 
We have the following four scenarios concerning satisfaction of traders:
\begin{enumerate}
\item $CR$ is good, $\VT$ is good too.
    
    In this case when everyone does her/his duties as expected and there is no dissatisfaction.
    
\item $CR$ is good, $\VT$ is bad.

Certainly, when collaborators are ready to commence their work, they anticipate the lock strategy to secure the coins in their memory. Failure to activate this lock strategy would raise suspicion among coworkers, potentially deterring them from initiating work as they would sense that something is amiss. Once it becomes evident that they have the expected coins in their memory, the collaborators initiate their tasks, and the investor possesses the appropriate coins, resulting in satisfaction for all parties involved.

Since both the workers and the investor are assumed to be reliable, satisfaction prevails throughout all rounds. Consequently, it becomes inconsequential whether the verification team requests confirmation from all the collaborators within the cooperation ring ($CR$). At the conclusion of the checkpoint, $\VT$ is responsible for transitioning the status of the "blocked" coins to either "paid" or "expired" and disbursing payment to the traders. Should $\VT$ fail to fulfill this duty, other verification teams in subsequent checkpoints will undertake this responsibility, as outlined in Rule~\ref{rule:VTConstraints}. 

Note that in this case and in an extremely worst case scenario, a bad verification team may disrupt the cooperation ring, resulting in dissatisfaction among the coworkers, potentially causing a reduction of at most $\pen$ of the coins within their accounts.

\item $CR$ is bad, $\VT$ is good.

When the verification team functions reliably, any sign of dissatisfaction expressed by even a single collaborator within a cooperation ring will promptly expose their untrustworthiness. 
In this scenario, again, for the worst case scenario, a bad verification team may disrupt the cooperation ring, resulting in dissatisfaction among the coworkers, potentially causing a reduction of at most $\pen$ fraction of coins within the trader's accounts.

\item $CR$ is bad, $\VT$ is bad too.

This scenario in the worst case, may result in $R$ rounds of dissatisfaction. This occurs because reliable traders expect the verification team to detect bad collaboration in the first round. However, as the verification process continues through subsequent rounds, it may fail to respond appropriately. In the worst case scenario, this can result in the complete loss of the trader's coins within the cooperation ring ($CR$).
\end{enumerate}



\begin{theorem}
    \label{thm:usersbenefit}
    For the worst case regime when existence of very large cooperation rings are assumed, if the penalty $\pen$ is proportional to the square of the inverse of the size of a cooperation ring, thenc
\end{theorem}
\begin{proof}
For the analysis, let us recall some facts as follows:

\begin{itemize}
    \item Since the probability of having a bad verification team is actually exponentially small and negligible (see Theorem~\ref{thm:Verifsound}), the worst case analysis reduces to computing an upper bound for dissatisfaction of a typical trader when he/she is in a bad cooperation ring verified by a good verification team.
    \item Based on our model in the proof of Theorem~\ref{thm:constantdgeree}, one may think of cooperation rings appearing independently at random in $\ell$ slices where $\ell$ is a constant upper bound, again justifying that the worst case analysis reduces to a guarantee of negligible probability of dissatidfiability in just one of the slices.
    \item It is well-known (see \cite{russia2,zolotukhin_nagaev_chebotarev_2018}) that for a random two-colored permutation on $N$ objects, cycles having $i$ white and $j$ black colors of constant length $i+j=k$, for an asymptotic regime when the limit of the fraction of while elements tend to a constant $1-\alpha$,  admits a Poisson distribution with parameter
      $$\lambda_{i,j}=\frac{1}{k}\binom{k}{i}(1-\alpha)^{i} \ \alpha^{j},$$
    and moreover, these events appear indipently for different pairs $(i,j)$.
\end{itemize}

Hence, fixing an integer $k > 0$ and a typical trader in an asymptotic regime, one may compute the probability of this trader appearing in a bad cooperation ring of length $k$ as

$$\displaystyle{\sum_{i+j=k, j \not =0}} \ \ \displaystyle{\sum_{k > 0}} \displaystyle{\sum_{t > 0}} \mathbb{P}[\#CR=t] \frac{kt}{N}= \displaystyle{\sum_{i+j=k, j \not =0}} \ \ \displaystyle{\sum_{k > 0}} \frac{k\lambda_{i,j}}{N}
=\frac{1-(1-\alpha)^k}{N} \leq \frac{\alpha \ k}{N}.$$
Consequently, a worst case upper bound for the expectation of the penalty fraction of this typical trader is equal to
$$\frac{\alpha}{N} \displaystyle{\sum_{k > 0}} k \ \pen.$$
Consequently,  if $k\pen = O(k^{-1})$ then
$$\lim_{N \longrightarrow \infty} \frac{\alpha}{N} \displaystyle{\sum^{N}_{k > 0}} k \ \pen \simeq \lim_{N \longrightarrow \infty} \alpha O(\frac{\log N}{N})  = 0.$$

 \end{proof}

 \begin{remark}
   \textbf{Comments on consequences of analysis in real implementations}.  
 \end{remark}


It must be noted that the assumptions of our theoretical analysis (and the proofs) are essentially tuned to be sufficient for a concrete result and are too strong compared to what happens in real applications. First, note that for a guarantee in Theorem~\ref{thm:Verifsound} it is sufficient that $\alpha=1-\epsilon$ be constantly bounded away from $0.5$ by some positive constant $\epsilon$ where the theorem provides a normal approximation. 

Also, it is instructive to mention that our worst case analysis in Theorem~\ref{thm:constantdgeree} is far from being sharp where in real applications one expects that the communication complexity (i.e. the maximum degree in the graph of $\LoR$) be controlled by an exponent much smaller than $\ell^4$. 
This shows that even in a worst case scenario, the communication and storage complexity for each trader is directly related to the number of cooperation rings they participate in. Therefore, traders do not require extensive storage capacity.

On the other hand, one may note that in real life, the size of cooperation rings is essentially constantly bounded (say by a constant $L$). It is interesting to note that, no rigorously, in this case based on results of \cite{benaychgeorges2009cycles,numofpermutation} one knows that the expected number of cooperation rings of size $k$ is approximately equal to $\frac{N^{k/L}}{k}$, implying that the probability that a trader falls into a bad cooperation ring of size $k$ is approximately equal to $f(\alpha) N^{k/L-1}$ for some function $f$.
Hence, the expectation of the penalty function $\pen$ in this case is bounded as follows,
$$\frac{f(\alpha)}{N} \displaystyle{\sum_{k > 0}}  N^{k/L} \ \pen \leq  \frac{f(\alpha) \pen}{N^{1-\frac{1}{L}}-1},$$
showing that one may choose the penalty function to be independent of the size of the cooperation rings (but definitely depending on the constant parameters of the system as the number of rounds) that provides more flexibility one the administrator side to set up an instance of $\LoR$.
Needless to say, from a theoretical point of view and as a byproduct, our contribution sets up a strong motivation for a rigorous analysis of {\it two-colored $A$-permutations} which seems to be open and remains to be considered in forthcoming research in random combinatorial structures.

\begin{remark}
\label{remark:typicalparameters}
    \textbf{On typical parameters of $\LoR$.}
\end{remark}
To sum up our analysis, let us comment on some typical parameters for $\LoR$. First, it ought to be mentioned that $\LoR$ is intentionally designed as a large-scale system to support a large number of traders in each one of its instances. Hence, let us assume that $N > 10^{5}$. Then for the security justification one may choose $250 \leq \kappa \leq 500$ assuming that $\alpha < 0.45$ guarantees a good bound for the security by Theorem~\ref{thm:Verifsound}. 

Also, for satisfiability one may choose the penalty function to be proportional to $1/R$ where $R$ is the number of rounds and make sure that it is proportional to the inverse of the size of the cooperation rings if there are some services that may admit quite large cooperation rings.

To determine the number of rounds, $R$, let $\gamma$ be the average submission rate of the cooperation rings in a fractal ring in the system. Also, let $C_{f}$ be the expected number of cooperation rings in a fractal ring, implying that the time required for a fractal ring to end is at most $\gamma C_{f}$. 
If the time required for a cooperation ring $CR$ is denoted by $\psi$, then the administrator may define the number of rounds, $R$, 
based on the duration of one round that can be determined as follows,
$$E[\psi]= A*R \qquad R=\frac{E[\psi]}{A} \qquad A= K*\gamma.$$
Here, the value $K$ is a regulator, and depending on the requirements of an instance of $\LoR$, one can set $K$ to provide a different round time. 

\section{Payment Rules, Policies of $\LoR$}
\label{policies}
$\LoR$ rules and policies are listed below:

 \begin{enumerate}
  \item
    \emph{Terminating suspicious cooperation rings}: 
    \label{rule1}
    
By the event of reaching the end of each round, each verification team checks if every cooperation ring in its fractal ring is verified by all the traders, based on a majority vote procedure.
Consider a specific round and a specific fractal ring. If at least one trader is not satisfied with the collaboration in a cooperation ring, that collaboration is not valid anymore, and such a cooperation ring will be terminated. 
Other cooperation rings are valid and their traders do not need to create another fractal ring, as the fractal ring was submitted successfully. 


\item 
\emph{Information of the cooperation rings}:
\label{rule:info}

A verification team should receive the required information (shared information on the checkpoints in the $\TCB$ and the $\RCB$ tables) of the cooperation rings from the traders of the corresponding fractal ring. To confirm the information, a random member of a previous verification team that was verifying the trader's older cooperation rings would be selected to pass the required information.

\item
 \emph{Incomplete cooperation}: 
\label{rule.wrong.round.policy}

Suppose cooperation $x$ in a fractal ring needs $R$ rounds to be accomplished. By the end of a round $i$, a trader in $x$ denoted by $t$ declines that the last round was successful. For the round $i$, the trader $t$ refuses to accept a successful transaction; the system decreases a percent of the payments of all the traders involved in that cooperation ring. This percent is specified in Theorem~\ref{thm:usersbenefit}. In fact, this is a punishment for an unfinished job. Note that $\LoR$ penalizes all the collaborators in an unsuccessful cooperation ring since there is no mechanism to distinguish the person who was trouble-making or to verify reasons that may disturb the communications between the traders. 



\item 
\emph{Payments at the end of checkpoints}:
\label{rule:payments}

 By the end of each checkpoint, the investors receive expired coins (see coin's status~\ref{cointable} for more detail). The expired coins reveal the services an investor received in the previous checkpoint. The system pays the traders who are investors based on their expired coins.
 The workers receive investment coins that reveal the payments of the workers at the end of checkpoints.
Note that verification teams carry the information about the investors and their expired coins, therefore no one can bring fake expired coins and lie about them.
Also, note that if the number of investors decreases gradually, the traders have to invest more to create more cooperation rings and help themselves to speed up their collaborations to get started sooner. Otherwise, there will be plenty of worker coins with a few investor coins. This fact keeps the system in a balanced configuration.

 

\item 
 \emph{Initiation}: 
 \label{rule:initiation}
 
 At the beginning, the system would be launched by one of the popular systems, namely Ethereum. We assumed that there are enough traders and different types of coins and cooperation rings in the system. Also, starting with a well-known cryptocurrency will incentivize the users of old systems to join the new ones.

\item 
\emph{Synchronization}:
\label{rule:sync}

The traders need to agree on a checkpoint and the administrator can use a well-mannered strategy called \emph{firing-squad} to implement a synchronization strategy. This strategy is useful at the start of the system for everyone to be synced on the first checkpoint. The \emph{firing-squad} strategy was first proposed by John Myhill in 1957 and published (with a solution) in 1962 by Edward Moore (for more on this see e.g. \cite{firingsquad1,firingsquad2,firingsquad3}).

\item 
 \emph{Entrance}: 
 \label{rule:entrance}
 
Any user who intends to participate in $\LoR$, and play a role as a new trader in an instance of $\LoR$ should buy enough amount of $\ARA$ from another already existing trader. This amount must equal at least to the cheapest service provided in that instance of the system. Their settlements are unrelated to policies of the $\LoR$ system. 

\item
\emph{Policies related to verification teams}:
\label{rule:VTpolicies}

Consider a situation where a trader $t$ participates in at least one cooperation ring of a fractal ring $\FR$. A new verification team $\VT$ is assigned to check $\FR$. 
One of the duties of the members of $\VT$ is to check if $t$'s opinion matches with the majority of the members of the previous verification team that $t$ plays a role in, and this verification team validates a successful fractal ring in a previously passed checkpoint.
If $t$ had confirmed some invalid cooperation rings, the trader cannot submit a fractal ring during the two subsequent next checkpoints of the system. The invalid conferment means the vote of the trader $t$ was not the same as the majority's vote. As we already know, the majority of the members of a verification team are reliable (see Theorem~\ref{thm:Verifsound}), and every trader tries to do the related duties correctly. 
In other words, every successful verification process passes through periodic checkings of cooperation rings. A cooperation ring must be checked during each round to be satisfied by its coworkers. Therefore, the aforementioned policy persuades traders in a verification team incentively check the status of the cooperation rings in a fractal ring and tell the truth.
Therefore, a trader $t$ should complete those duties as a verification team member in checkpoints and through system rounds.

{\bf \small Submission process}: 

Each fractal ring needs to be submitted to the system's tables by a designated verification team. A typical fractal ring, denoted as $\FR_{(\VT)}$, is assigned to the verification team $\VT$. Assuming that the current system state is at checkpoint $\CP_{p}$ and the subsequent checkpoint is $\CP_{p+1}$:

Let's consider a scenario where there are currently a total of $j$ requests for coins being broadcasted within the system before reaching checkpoint $\CP_{p}$, followed by $k$ requests for coins broadcasted after checkpoint $\CP_{p}$.
The fractal rings produced by traders are disseminated throughout the system and will be taken into account during the verification process conducted by their respective verification teams at the upcoming checkpoint.
The verification team, referred to as $\VT$, is tasked with confirming the validity of each cooperation ring denoted by $CR$ within the corresponding fractal ring $\FR$ based on the following criteria:
\begin{enumerate}
    \item 
The coins present in $CR$ must have been broadcasted within the system prior to the occurrence of checkpoint $\CP_{p}$.
\item
Trader $t$ must have a sufficient account balance before checkpoint $\CP_{p}$ to initiate a request for a specific coin ($c$).
\item
The verification team is responsible for verifying the reduction of funds from trader $t$'s account and the subsequent creation of the coin table for $t$. The verification process extends to validating all entries within the coin table created for each coin.
\item
The cooperation table corresponding to each $CR$ is subject to verification to ensure the accuracy of its entries. An example includes verifying the precise calculation of a cooperation ring's weight.
\end{enumerate}


If $\VT$ confirms $\FR$, it is ready to be submitted in the tables of the system, and its cooperations are allowed to start their work by the arrival of $\CP_{p+1}$. Note that the time $\VT$ can spend to do the verification tasks related to the cooperation ring $CR$ is less than or equal to the time between $\CP_{p}$ and $\CP_{p+1}$. If $\CP_{p+1}$ gets passed, all the above-mentioned jobs should get postponed to the next checkpoint. 
If $\VT$ does not confirm one of the above-mentioned cases, then $\FR$ cannot get submitted. 
The default policy of the system is to remove the information of $\FR$, and the  trader $t$ should submit another fractal ring for the next checkpoints. 


Following the commencement of activities at the conclusion of the subsequent checkpoint, it is expected that the cooperation rings will conclude their tasks.  
 The verification team $\VT$ compensates traders who successfully complete their collaborations, and this involves updating the account numbers and corresponding entries within the $\RCB$ and $\TCB$ tables.

\item \emph{Constraints of verification teams}:
\label{rule:VTConstraints}

A verification team $\VT$ assigned to a fractal ring $\FR$ has access to only parts of the information of the traders. See the following rules. 
\begin{itemize}
\item 
Traders' information is duplicated across multiple traders' memories. However, except for the current information regarding $\FR$, including associated coins and cooperation rings, all other data is read-only for $\VT$. This implies that $\VT$ lacks the authority to delete accounts or arbitrarily modify the account balance. The role of a verification team is limited to adding to a trader's account based on the \textit{blocked} coins recorded in the trader's memory at the end of a checkpoint. 

\item 
If a previous verification team from a past checkpoint has not fulfilled its obligations to pay the traders, and there are some blocked coins left in the memory of a trader, the current verification team $\VT$ will settle these coins by adding their value to the accounts of the respective traders (the status of the coins will be changed to paid). In essence, at the end of a checkpoint, a verification team relies on the presence of blocked coins in a trader's memory to determine the amount owed to them.

\end{itemize}
 



\item 
\emph{Earning ARA}:
\label{rule:earningARA}

In the $\LoR$ system, a user can only receive some $\ARA$ from an already member of the system who owns some money (probably with an outline settlement), or the user can participate in a cooperation ring and work or invest in an actual collaboration and receives some $\ARA$. 
This guarantees that one cannot use $\LoR$ to earn money effortlessly. 


\item
 \emph{How does the system support long-term cooperation?}
 \label{rule:longtermcoop}

After passing a checkpoint, based on the checkpoint time period (by the end of a checkpoint process), the $\TCB$ table will be updated, and the system needs to pay the traders who have at least one round of a successful cooperation ring, and their remaining rounds equals zero. If there still remain some rounds required for a cooperation ring in a fractal ring to get finished, then the payment of that ring should be postponed to the next checkpoint.


\item
\emph{Fractal rings Must be valid on the very first try}: 
\label{rule:firstFR}

When a verification team verifies a fractal ring, if even one cooperation ring is invalid or anything else is incorrect, the fractal ring is not valid, and the trader who sends it must resubmit another fractal ring for the next checkpoint.
Hence, everyone should submit a valid fractal ring on the very first try. 
This makes sure that traders have to pay for the initial coins. As a result, the traders trying to cause harm actually paid for coins in the system, and now their coins are locked. This also makes it even better for other traders because they are less likely to encounter harmful traders.

\end{enumerate}

\section{Discussion}
\label{sec:discussion}
This paper introduces a novel middleware system that employs randomization to establish a dependable environment for user collaboration, eliminating the need for substantial financial investments or high-performance hardware. Notably, the system operates consistently over time, maintaining its efficiency. This fosters increased user trust and reliable cooperation. The system proves particularly advantageous for enhancing security in freelancing projects and managing Internet of Things (IoT) systems.

$\LoR$ employs two types of rings with a randomized selection structure. These rings continuously initiate and conclude events, ensuring the verification and validation of collaborations.



 \bibliographystyle{ACM-Reference-Format}
\bibliography{sample-base}
\appendix



\section{Propoerties of $\LoR$}
\label{sec:properties}
Here, we outline some key properties of the $\LoR$ system.

\begin{enumerate}
\item
\emph{Mutual exclusion}: 
 \label{rule:mutualexculsion}
 
No trader can submit an already submitted cooperation ring. When a fractal ring is submitted in the tables of the system, all of its cooperation rings are submitted. The tables of the system are mutually excluded, so no two cooperation rings may get submitted simultaneously. $\TCB$ and $\RCB$ tables must be implemented to be mutually excluded for the traders when they want to submit a fractal ring.

 \item
\emph{Fake functions}:
\label{rule:functions}

Traders create coins, cooperation rings, and fractal rings using specified procedures within the context of $\LoR$.
Verification teams are responsible for checking the behavior of traders and their proposed structures to ensure that fake functions are not used.
 So, if a trader uses a fake function, it will be detected by the verification teams. 

 \item 
 \emph{Traders are motivated to work with only one account}: 
 \label{rule:oneaccount}
 
A trader who owns more blocked coins receives more extra money ($\ARA$). This extra money should be considerable to make the traders be incentive to work with only one account.

\item 
\emph{Incentives to create fractal rings}:
\label{rule:FRincentives}

The identity of the trader initiating a fractal ring is inconsequential, as all traders can place their confidence in the verification teams (as indicated in Theorem~\ref{thm:Verifsound}). The increased creation of fractal rings involving traders leads to an accelerated commencement of their respective tasks. Consequently, all traders share an incentive to engage in the creation of fractal rings.

\item
\emph{The probability of a trader joining a verification team is the same for all the traders:}
\label{prop:eaualVTprobability}


Verification teams are chosen according to the uniform distribution on the set of traders with replacement. This ensures that, asymptotically, each trader is expected to appear equally likely as a member of a verification team.



\end{enumerate}
\subsection{Some particular charactristics of $\LoR$ }
\label{subsec:characteristics}
 The spirit of all the previously collaborative distributed systems might be different from the $\LoR$ system in the below-mentioned features: 
\begin{enumerate}

 \item {\bf Security and reliability}:
 \label{prop:security}
 The decentralized random structure of the procedures, as well as an intentional lock mechanism used in $\LoR$, makes this platform highly secure against malicious traders. In fact, $\LoR$ incentivizes honesty. The system employs strategies to motivate truthful job performance.
 
 $\LoR$ is secure from many attacks. To engage in collaboration within the $\LoR$ system, a user needs to invest by purchasing an adequate amount of the internal cryptocurrency employed by the system. This investment serves as a preventive measure against numerous types of attacks. This requirement differs from many previous systems where such an investment is not mandatory.
    
Moreover, due to Theorem~\ref{thm:Verifsound} the verification teams are highly reliable. A uniform distribution chooses each trader with the same probability of being among a verification team. So, considering a long period of time, the traders are participated in the same number of verification teams. The behavior of a verification team member will be checked during the following checkpoints of the system in such a way that it would be better for them to do their duties correctly.

See Section~\ref{sec:attacks} for more information on examples of how the system prevents potential attacks.
 
 \item {\bf Cooperation vs. Transaction}:
 \label{prop:coop_vs_tran}
 Previous systems, such as Algorand, focus on transactions.
Cooperation in the $\LoR$ system holds broader significance compared to transactions. While transactions involve exchanges between two individuals, cooperation encompasses interactions among multiple (probably more than two) users. It extends beyond mere money transfer and involves the shared effort of collaborators to complete a task over a specific period. In this context, system security goes beyond financial aspects. It centers on ensuring the satisfaction of all traders with their respective tasks and maintaining their efficient functionality.

\item {\bf Ensures fairness (No starvation)}:
Slower users are not left disadvantaged.
Each trader has the capability to initiate coin creation requests, and this process is independent of the functionality of other traders. This means that every trader can create coins, which are then shared with the system. The system, in turn, can combine these coins to form cooperation rings. The true competition among traders centers around transforming these cooperation rings into fractal rings. In essence, traders aim to secure more coins to establish both cooperation and fractal rings, thereby expediting their own work processes. Consequently, even a trader with slower coin creation can still collaborate effectively with other traders. Furthermore, due to the consistent distribution of random processes in $\LoR$, coins of the same type possessed by different traders have an equal likelihood of being incorporated into cooperation and fractal rings.

\item {\bf No collusion}:
As mentioned previously, the system assigns coins to each other in a random manner. This means that traders cannot select or pick their collaborator's coins. This approach ensures that there are no expected or pre-arranged co-workers, and it eliminates the possibility of collusion among traders.

\item {\bf No inflation}: 
\label{chr:inflation}
$\LoR$ intends to keep the primary currency ($\ARA$) valuable in time.
Considering a fractal ring, if we know that equal or less than $R$ rounds were successful, the verification team will pay a percentage of each trader's fee; for example, ninety percent. The admin could adjust the value of this fee at the system's initiation step or by a fixed procedure that uses some computations. These computations are based on the application and should ensure that by the passage of time in the system, the inflation of the primary currency gets around zero.


\item {\bf No need to detect offenders}: The $\LoR$ system doesn't identify individual offenders; rather, it penalizes all users involved in an unsuccessful cooperation ring. This approach is elaborated upon in the paper, where system policies utilize this method to maintain trustworthy connections between users and hold accountable those who exhibit unfavorable communication or cooperation behavior. This mechanism is detailed in Payment and Policies, rules~\ref{rule1} and~\ref{rule.wrong.round.policy}.

\item {\bf Versatility}: 
The $\LoR$ system serves as a versatile and general-purpose platform. It is not tailored to a specific function, be it IoT management or cooperative computing, or any other specific purpose. The $\LoR$ system has the capability to accommodate any service that satisfies the criteria outlined in subsection~\ref{subsec:coins}.

\item {\bf Being balanced}:
\label{char:balanced}
Traders have a strong incentive to work together across different services. When the number of requests for a particular service becomes substantial, traders find it more advantageous to collaborate on additional services. This is the reason we anticipate a significant portion, if not all, of the traders to be capable of providing a diverse range of services.


\item {\bf Working legitimately}: In the $\LoR$ system, every trader is required to actively participate and contribute to earn money. Users must either offer a service or invest in one. The system ensures that no energy or effort is wasted – every action has a purpose.
Note that traders without any coins won't gain any advantages from submitting fractal rings.
 This encourages traders to actively engage in the system, generate additional cooperation rings, and increase their earnings.

\item {\bf Fractal structure}:
The $\LoR$ system gains an advantage from its fractal structure, which greatly boosts its flexibility. This design not only simplifies system recovery but also streamlines the process. It's noteworthy that any collaboration ring can seamlessly fit into a fractal ring, effectively acting as a fractal within that larger structure. This arrangement aids in distributing verification teams evenly across all trades. Moreover, the system can rapidly generate new fractal rings whenever the need arises.

\item {\bf Motivating users}:
The $\LoR$ system does not deal with the issues due to the lack of motivated users. This problem has also been studied by \cite{brown2018formal,carlsten2016instability,kiayias2016blockchain}. A trader earns more money than expected. The $\LoR$ system considers an extra $\ARA$ for a trader after the completion of a cooperation ring. This motivates the trader to work more in the system. On the other hand, the members of verification teams are motivated to do their duties correctly to submit fractal rings. 
These and some other policies of the system motivate users to work on the $\LoR$ platform and be honest.

\end{enumerate}


\section{The $\LoR$ system in details}
\label{sec:system}
This section covers the details of each entity of the $\LoR$ system.
\subsection{Administrator}
\label{subsec:admins}
In every instance of the $\LoR$ system, there exists a single administrator. 
Aside from the selection of services and their respective prices, an administrator has the option to modify the infrastructure of the instance utilized for an application.
Additionally, the administrator holds the authority to introduce alternative methodologies in implementing specific facets of the system. To illustrate an example, the administrator may opt for a strategy centered around reputation when appointing verification teams. It is important to note, however, that the $\LoR$ system solely guarantees the effectiveness of the methods initially delineated within its original specifications.

   Note that the specified service prices should be reasonable and align with both the application's requirements and user expectations. This equilibrium prevents scenarios of excessive funds that discourage participation or insufficient funds that compromise the system's robustness. To manage factors like inflation, the system retains the ability to adjust payments (refer to Subsection~\ref{subsec:characteristics}, 
   Characteristic~\ref{chr:inflation}). So, in the $\LoR$ system, the prices remain constant by the passage of time.
\subsection{Traders}
\label{subsec:traders}
Each trader possesses an account, a wallet, a distinctive ID, and various other details, all of which are stored within an entry in the $\TCB$ table, a topic we will delve into subsequently. The wallet consists of both a private key and a public key, operating in the same manner as the wallets employed in earlier systems~\cite{cite:mining}.

A user must obtain a sufficient amount of internal crypto-currency (ARA) from an established trader within the system. The minimum required ARA for a trader to participate in a cooperation ring corresponds to the cost of the most affordable service, as specified by the administrator. 

\subsection{Coins}
\label{subsec:coins}
Traders have the ability to contribute to various types of services. These services are represented by units known as "coins" within the $\LoR$ system. 

In order for a service to be approved for integration into the $\LoR$ system, it must meet the following specified criteria.

{\bf Requirements for a Service to be a \emph{Coin}}:
\label{servicerequirements}
\begin{enumerate}
    \item It needs to meet the qualification for propagation through the Internet.
    \item It needs to be measurable and quantifiable. This implies that a coin should have a fundamental unit, allowing any request to be represented as a multiple of that unit. This division should consider factors like qualification level, time allocation, and resource allocation. For instance, a trader might contribute computational power as a CPU-coin through time-sharing.
\item 
The majority of traders should have the ability to offer this service to others. In other words, a service defined as a specific type of coin should be a universal capability accessible to all (or the majority) of the traders (see Characteristic~\ref{char:balanced} to know about the reason).
\item $\LoR$ has some running rounds. A coin type should have the capability of getting split into as many rounds as it takes for a cooperation ring to get finished. 
Imagine that, for instance, after three rounds, the connection between the involved traders is unexpectedly disrupted. In this situation, the system should be able to calculate the payments owed to the traders who contributed in the preceding three rounds. Therefore, a service tied to a coin's definition must be designed in a way that allows for flexible division into rounds. Each round should hold substantial value to warrant payment.
\end{enumerate}

In each instance of the $\LoR$ system, these coins are categorized into distinct classes. 
Among these classes, there exist specific types referred to as ``investment coins''. When traders request investment coins, they are signifying their intention to participate in an investment for the initiation of a collaborative endeavor.

As previously explained, each coin is represented by a distinct structure referred to as a  "coin table' (refer to Figure~\ref{fig.Coin.Structure}). The coin table comprises a single column and several rows, with each row representing a specific specification. This arrangement serves to establish connections among collaborators, forming cooperation rings.

In practice, the coins owned by collaborators engaged in a cooperation ring are subjected to a locking mechanism, which will be further elaborated upon in Subsection~\ref{subsubsec:lock}.

\setlength{\textfloatsep}{5pt}
\setlength{\intextsep}{10pt}
\begin{figure}[ht]
\begin{center}
\includegraphics[width=0.6\textwidth]{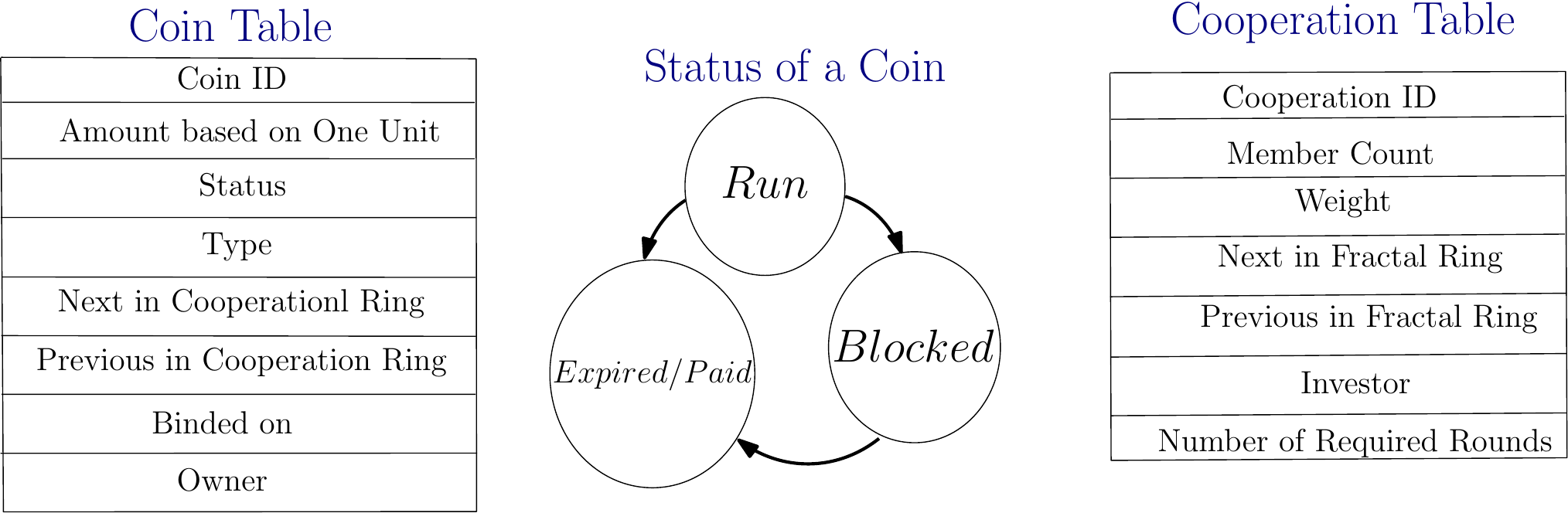}
\caption{{\small The coin table and the cooperation table. The left table includes all information required for each currency. The right table illustrates the cooperation table, which includes the information required for a cooperation ring in a fractal ring. }} 
\label{fig.Coin.Structure}
\end{center}
\end{figure}

As previously noted, the $\LoR$ system operates with a primary currency referred to as $\ARA$, which determines the cost value of each coin. The value of one unit of this currency is represented by $\mathcal A$.

In every cooperation ring, there exists a designated trader known as the "investor." To assume this role within a cooperation ring, a trader must request investment coins. When trader $t$ makes a request for a specific type of coin, the system employs a standardized function $f$ to generate a corresponding coin table based on trader $t$'s request (refer to Rule~\ref{rule:functions} for information about how these functions are secure).

The entries of a coin table are explained below; 
\begin{enumerate}
\label{cointable}
 \item 
    \emph{Coin-ID}: Every coin has a unique ID. This ID will be generated by $f$. 
\item \emph{Amount-based-on-One-Unit}: 
This entry discloses the price of the coin in terms of $\ARA$, one unit of $\ARA$ is denoted by $\A$, for example, $0.45\A$.
   \item 
   \emph{Status}: 
This entry has four possible values: Run, Blocked, Expired, and Paid.

1- If a coin has been used in a previous cooperation ring, its status becomes \emph{Blocked}.

2- When cooperation concludes and a coin resides in the memory of an \emph{investor}, its status changes to \emph{Expired}.

3- If the coin, which is owned by an investor and stored in the memory of other traders as their payment, its status changes to \emph{Paid}.

4- When a coin is available for a trader to associate with a group of other coins, its status should be \emph{Run}. This means it is a newly created coin based on a request.

At the conclusion of a checkpoint, the verification teams remunerate the traders based on the status of the coins, as indicated through this record.
    \item  \emph{Type}: 
   This record indicates the category of a coin. It cannot be changed once entered. Importantly, traders are not allowed to modify or delete any coin category within the system. Each coin category should have a well-defined \emph{exchange rate} relative to $\ARA$.
    \item  \emph{Next/Previous-in-cooperation ring}: 
A collection of colleagues needs to be connected within a cooperation ring. This ring is established based on the coins rather than the traders themselves. A trader has the option to participate in multiple cooperation rings by utilizing different coins. Two connections within coin tables play a role in forming these cooperation rings.
    \item  \emph{Bind-on}: 
   When a coin is held in another trader's memory, this record indicates whose memory it is stored in. The $\LoR$ system employs a method to establish an artificial lock. This approach can be thought of as a simulated deadlock, where traders wait for each other until a checkpoint is reached (refer to Subsection \ref{subsubsec:lock} for more information). This record serves the purpose of that strategy. The locking strategy mandates that traders who seek investment must pay first before receiving any services. Additionally, this strategy guarantees that all working traders complete their tasks accurately before receiving payment.
    \item \emph{Owner}: The owner of a coin is the trader who asked for the coin before.
\end{enumerate}

\subsection{Cooperation Rings}
\label{subsec:cooperationrings}
A group of traders, selected at random for collaboration, is portrayed within a cooperation ring structure. To introduce randomness, a viable approach could mirror the technique employed in \cite{gilad2017algorand}, which relies on verifiable-random-functions \cite{micali1999verifiable}.
Every cooperation ring possesses a unique organizational layout termed a \emph{``cooperation table"}, as depicted in Figure~\ref{fig.Coin.Structure}. This cooperation table empowers a trader to establish a fractal ring by picking a suitable subset of the available cooperation rings within the system. To generate a cooperation table for each cooperation ring, a trader is required to execute a specific function.
Traders who partake in a cooperation ring are limited to initiating their activities from a designated checkpoint established by the system (the corresponding verification team).


\subsubsection{\bf The Lock Strategy}
\label{subsubsec:lock}
Every coin that a trader uses within a cooperation ring is linked to the memory of another trader. The state of these coins undergoes a transition from \emph{Run} to \emph{Blocked}.
This deliberate locking mechanism creates a situation where coworkers are awaiting one another. The system, operated by a verification team, enforces a blockage on coins within a coworker's designated memory. This blockage persists until the cooperation concludes or until the system forcefully terminates the cooperation. This mechanism serves to effectively enforce dependable payment protocols within the context of the $\LoR$ system (see Section~\ref{policies} for the payment policies). 
Traders will not receive payment or gain any benefits from their collaborations until all the trades within a cooperation ring have been fulfilled to the last successful round.
Observe that the structure of a cooperation ring establishes simple links that enable coins to be easily accessed by verification teams.

Furthermore, these verification teams can retain relevant coins and cooperation ring information, enabling them to identify any intentional fraudulent modifications. This heightened security measure contributes to the overall robustness of the system.

\setlength{\textfloatsep}{5pt}
\setlength{\intextsep}{0pt}
\begin{figure}[ht]
\begin{center}
\includegraphics[width=0.7\textwidth]{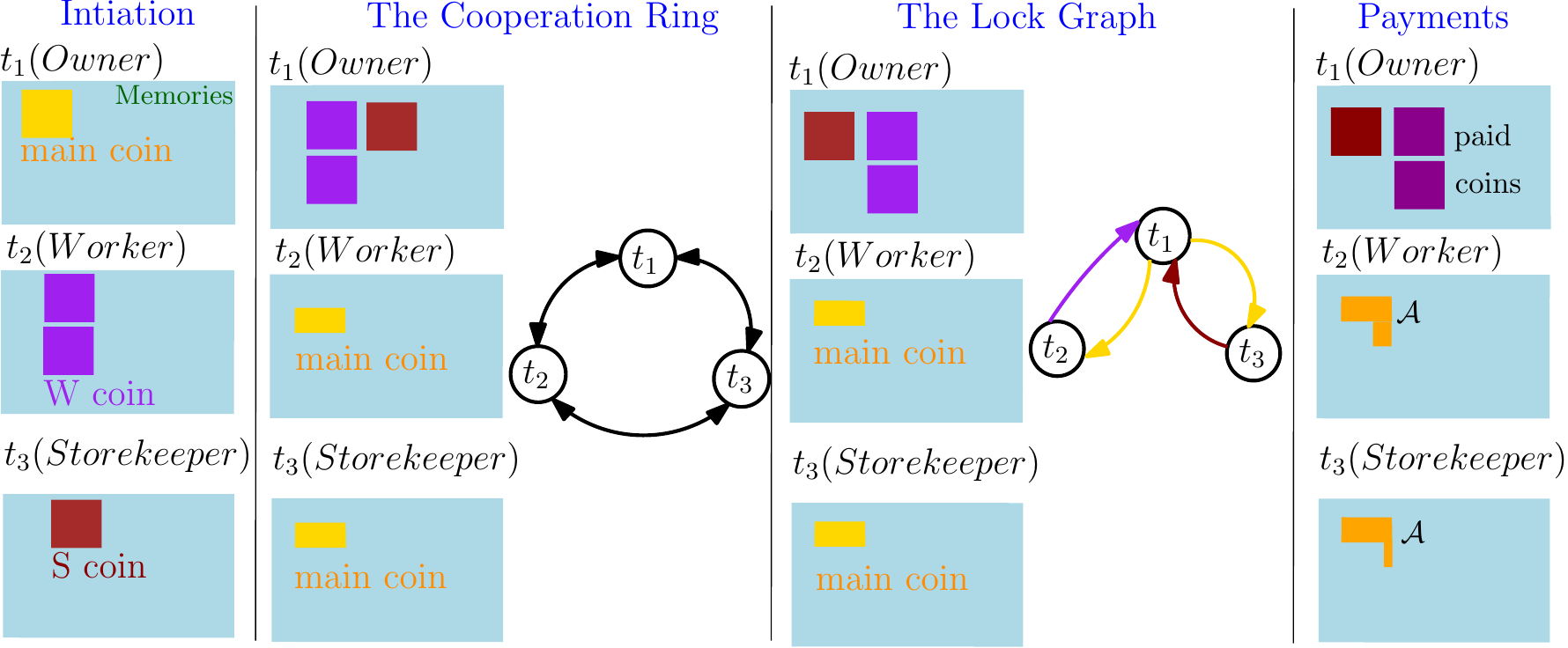}
\caption{
{\small An example to show how we use an intentional lock to make sure no one can access used coins in their own cooperation rings without the system's permission. Traders will receive some extra money ($\ARA$) based on their contributions. Trader $t_{1}$ owns one \emph{investment coin}, trader $t_{2}$ has two $W$-Coin, and trader $t_{3}$ has one $S$-coin in their memories individually. These coin names stand for a calculating work ($W$-coin) and provide enough space as storage ($S$-coin). The system has already selected these three traders randomly to be in a group for cooperation. 
In the lock graph, $t_{1}$ must pay $t_{2}$ and $t_{3}$ based-on their contribution. In this group $t_{3}$ and $t_{2}$ are waiting for $t_{1}$ to pay them. 
The trader $t_{1}$ is waiting for $t_{3}$ and $t_{2}$ to provide its requirements resource.} }
\label{fig.lock.example}
\end{center}
\end{figure}

An example of a lock graph is illustrated in Figure~\ref{fig.lock.example}. 
Through the intentional locking mechanism, the status of each coin is changed to \emph{Blocked}. To expedite the payment process, verification teams modify the coins (possibly with a partitioned amount) and relocate them to traders' memories based on the lock graph (refer to Figure~\ref{fig.lock.example}). However, if a cooperation ring terminates, the payment process is handled differently. Please refer to Section~\ref{policies} for related policies.
As previously discussed, the \emph{Bind-on} entry within a coin table must be configured with the trader-ID of the trader responsible for holding that coin within their memory.
Modifications to the coin table for a specific coin are solely permissible through a decision reached by the majority of verification teams.

The comprehensive depiction of entries within the cooperation table (refer to Figure~\ref{fig.Coin.Structure}) is elaborated
below. 
\begin{enumerate}
    \item Cooperation ID: Each cooperation ring is assigned a distinct and universally unique identifier (ID) that remains consistent across the entire $\LoR$ system's lifespan.
    \item Member Count: This parameter denotes the number of traders participating within a cooperation ring.
    \item Weight: This is the summation of all coins amounts involved in the cooperation ring, based on $\ARA$.
    The amount of a coin demonstrates its cost value in the system and can be calculated based on its amount and type.
    The \emph{Weight} of a cooperation ring is calculated based on the cost values of all of its coins. 
    This entry can help the verification team calculate the budget that the system should consider for the cooperation ring as payments.  
    \item Next-in-fractal-ring: Indicates the next cooperation ring in a fractal ring assembled by a trader.
    \item Previous-in-fractal-ring: Indicates the previous cooperation ring in a fractal ring assembled by a trader.
    \item Investor: This is the ID of the trader who invests to start a cooperation ring. 
\end{enumerate}

$\star \ \ $ Each trader must use an algorithm to assign coins to a cooperation ring. This algorithm follows a clear process to match coin types within the ring correctly. To bring randomness to its choices, the algorithm uses a special function called $\SHA_{256}$, which is like a digital dice roll. By using $\SHA_{256}$, the algorithm makes sure its choices are random. Importantly, the outcome of $\SHA_{256}$ can be checked to fit the cooperation ring's context, preventing anyone from dishonestly setting a fixed assignment. The pseudo-code outlining the mechanics of this algorithm is illustrated in Algorithm~\ref{algo.3}. In this pseudo-code, $F_{i}$ is a regulating function that creates a cooperation ring out of the result of $\SHA_{256}$. This function generates a suitable cooperation table and other corresponding requirements. Using the investment requests and their corresponding investing coins as a guide, the algorithm randomly finds the right coworkers for collaboration.

\setlength{\textfloatsep}{5pt}
\setlength{\intextsep}{10pt}
{\small
\begin{algorithm}[H]
\caption{Algorithm to create cooperation rings}
\begin{algorithmic}[2]
\Procedure{Randomized}{}
\State define

$\mathcal{CR}$ = The output cooperation ring.

$M$ = A  selected  investment coin.

$|M|$ = The number of required coins.

$F_{i}$ = A deterministic algorithm to assign the right type of coins to $\mathcal{CR}$.
\For{ i= 2 ; $i \leq  |M|$; i++}
    \State $\mathcal{CR}_{[i]}  = F_{i}(\SHA_{256}(M))$  
\EndFor
\State Return $\mathcal{CR}$
\EndProcedure
\end{algorithmic}
\label{algo.3}
\end{algorithm} }

 
\subsection{Fractal Rings}
\label{subsec:fractalring}
A fractal ring is a block of cooperation rings created by a trader and verified by a verification team. A trader $t$ should assemble a random number of
cooperation rings to create a fractal ring (depicted as Fractal Rings in Figure~\ref{fig.overal}). 
This ring arrangement offers a straightforward approach for accessing and verifying cooperation rings. The precise quantity of cooperation rings needed to construct a fractal ring is subject to randomness. 
A cooperation ring is a fractal element of a fractal ring; any cooperation ring that exists in the system can be assembled for a fractal ring to be created. 
As the construction of a fractal ring lacks deterministic predictability, this characteristic prevents any exploitation of the system's rules.
To elaborate, if the number of cooperation rings within a fractal ring is mandated to be a fixed quantity, it could lead to a competitive race between the traders to gather the terminative cooperation rings for their fractal rings. 
Also, with a fixed number of fractals, the malicious traders might request a large number of cheap coins and request low-weighted cooperation rings, and the high-weight coins and cooperation rings might starve. 

Algorithm~\ref{algo.4} presents a pseudo-code that outlines the process of determining the number of cooperation rings needed within a fractal ring for it to be deemed valid, thereby allowing a trader to submit that fractal ring. In this algorithm, $EC_{f}$ denotes the expected number of required cooperation rings to submit a successful fractal ring. The deterministic function $F_{1}$ uses a random parameter to use $\SHA_{256}$ more securely. For this function, we assumed for $EC_{f}$ to be $500 \leq EC_{f} \leq 2000$. The function creates a uniform random number between $500$ and $2000$. Consequently, $EC_{f}$ coordination rings can merge into a fractal ring.

\setlength{\textfloatsep}{5pt}
\setlength{\intextsep}{5pt}
\begin{algorithm}[H]
\caption{Randomized Algorithm to find the number of cooperation rings required to submit a fractal ring}
\begin{algorithmic}[4]
\Procedure{Randomized}{}
\State define $EC_f = 500$, $H = Random \ number$; \small{Set the minimum number of cooperation rings in the fractal ring to be $500$.}
\State $FC_f = F_{1}(\SHA_{256}(H,Fractal[i]))$
\State Return $EC_f$
\EndProcedure
\end{algorithmic}
\label{algo.4}
\end{algorithm} 

Algorithm~\ref{algo.2} reveals how a trader $t$ can create a fractal ring denoted by $\FR$. 
The function denoted as $F_{2}$ constitutes the secondary random function, which is responsible for associating a Cooperation-ID with a randomly generated number produced by $\SHA_{256}$ (note that $\SHA_{256}$ is employed in the algorithms to establish uniform random functions). Remember that once a cooperation ring has been submitted, it is precluded from being submitted in any other fractal ring.

\setlength{\textfloatsep}{5pt}
\setlength{\intextsep}{5pt}
\begin{algorithm}[H]
\caption{Randomized Algorithm to submit a fractal ring}
\begin{algorithmic}[2]
\Procedure{Randomized}{}

$|\mathcal{FR}|$ = \small{The number of cooperation rings required to create $\FR$}.

$CR_{i}$ = \small{The $i$th cooperation ring of $\FR$}.

$CR_{1}$ = \small{A cooperation ring which is assembled by $t$ to start $\FR$}.

$F_{2}$ = \small{A deterministic algorithm that assigns the Cooperation-IDs to $\SHA_{256}(CR_{i-1})$. Based on the data of a cooperation ring, $F_{i}$ sets another cooperation ring}.

\For{ i= 2 ; $i \leq  |\mathcal{FR}|$ ; i++}
    \State $CR_{i}  = F_{2}(\SHA_{256}(CR_{i-1}))$ 
\EndFor
\State Call Algorithm~\ref{algo.1}
\State Return $\FR$
\EndProcedure
\end{algorithmic}
\label{algo.2}
\end{algorithm}

Algorithm~\ref{algo.2} calls Algorithm~\ref{algo.1} to establish a verification team specifically for the recently formed fractal ring denoted by $\FR$. 
This team consists of $\kappa$ members, with the selection of these members being conducted randomly from the entire pool of traders within the system. The obligations and functions of this verification team are expounded upon in detail in accordance with Rule~\ref{rule:VTpolicies}. The third random function of the $\LoR$ system is used by Algorithm~\ref{algo.1}.
It is important to highlight that every trader has an equal likelihood of being assigned to a verification team (refer to Property~\ref{prop:eaualVTprobability}).

\setlength{\textfloatsep}{5pt}
\setlength{\intextsep}{5pt}
\begin{algorithm}[H]
\caption{Selecting members of a Verification Team for a fractal ring $\FR$}
\label{alg1}
\begin{algorithmic}[1]
\Procedure{verification team}{}
\State define  $\kappa =  503 \leq Random \ Number \leq \ the \ number \ the \ members \ of$ $\FR$;
\State define $\VT_{i}$ = The $i^{th}$ member of the current verification team.
\For{ i= 1 ; i <=  $\kappa$; i++ }
    \State $\VT_{i}  =$ $\mathit{\SHA}_{256}(i,Users[i]\ in \ \VT)\%\kappa$\ 
\EndFor
\State $i$++
\State Return $\VT$
\EndProcedure
\end{algorithmic}
\label{algo.1}
\end{algorithm}


\subsection{Ring Control Block and Traders Control Block}
\label{subsec:TCBRCBCP}
The details of each fractal ring's information are stored in a table labeled as $\RCB$. Each cell within the $\RCB$ table corresponds to a specific cooperative ring within a fractal ring.
Every trader involved in a fractal ring, along with the traders participating in the associated verification team, can possess a duplicate of this table.

Similar to a cooperative ring, a fractal ring possesses a distinct universal identification that can be traced within the $\TCB$.
The coin table, cooperative table ($\RCB$), as well as traders' details encompassing their financial standing, cooperative rings, initiation point, and conclusion point of a cooperative ring (as per the round number designated by the system), are all stored within the trader's memories. As a result, multiple copies of this data are present in traders' memories.
The aggregate compilation of these individual cache memories defines the $\TCB$. Each trader retains a subset of this database; however, this subset solely comprises the pertinent information required by the trader.
This database must be regularly updated whenever a new fractal ring undergoes verification (submission) by a verification team.

Each verification team evaluates the condition of its associated fractal ring. This evaluation encompasses various aspects, such as calculating the payments due to traders up until the most recent checkpoint. For additional information, please refer to the policies outlined in Section~\ref{policies}.

For the implementation of $\TCB$, one viable approach is to utilize a distributed storage platform like Swarm\footnote{\url{https://ethersphere.github.io/swarm-home}}~\cite{buterin2013ethereum} within the Ethereum framework.
It is important to highlight that the system administrator must devise a solution to ensure mutual exclusion for both the $\RCB$ and $\TCB$ tables. Enforcing mutual exclusion for storage within distributed systems has posed a longstanding challenge. However, this issue has been tackled in previous research, as evidenced by references such as~\cite{ramachandran1994distributed,taubenfeld2004black}.


\section{Security and Attacks}
\label{sec:attacks}
In this section, we will establish the meanings of attack, security, and reliability. Subsequently, we will outline several reasons illustrating why $\LoR$ stands out as a secure and reliable system.
\begin{definition}[\emph{Attack}]
\label{def:attack}
   
An attack refers to a deliberate and malicious action taken to compromise the confidentiality, integrity, or availability of a computer system or its data. Attacks can take various forms, such as viruses, malware, phishing, denial-of-service (DoS) attacks, and more. Attackers may aim to steal sensitive information, disrupt services, gain unauthorized access, or cause other types of harm to the targeted system.
\end{definition}

\begin{definition}[\emph{Security}]
\label{def:security}
Security involves implementing measures and protocols to protect a system from unauthorized access, data breaches, and various types of attacks. It encompasses a range of practices, technologies, and strategies aimed at safeguarding computer systems, networks, and data. Security measures can include firewalls, encryption, access controls, intrusion detection systems, regular updates, and user training to prevent and mitigate potential threats.
\end{definition}
\begin{definition}[\emph{Reliability}]
\label{def:reliability}
Reliability refers to the ability of a system to consistently perform its intended functions without failure or interruption. A reliable system is one that can deliver its services and handle tasks as expected, maintaining its performance and availability over time. Achieving reliability often involves redundancy, fault tolerance, backup systems, and effective maintenance practices.
\end{definition}
It is important to note that security and reliability are closely interconnected. A secure system is often more likely to be reliable, as security measures can help prevent disruptions caused by attacks. Similarly, a reliable system can contribute to better security, as an unstable or unreliable system may be more vulnerable to exploitation by attackers.

In the realm of computer systems, achieving a balance between security and reliability is crucial for maintaining the functionality and integrity of the systems and the data they handle.
In the sequel we explore several attack scenarios and how the policies of $\LoR$ serve to counteract potential threats. For a broader perspective on different types of attacks in distributed or decentralized systems, refer to \cite{citeus,li2017survey}.

$\LoR$ enforces a set of policies that play a pivotal role in securing the system as follows:

\textit{Uniform Verification Team Selection}: Members of verification teams are chosen at random from the entire pool of traders.

\textit{Optimized Round Timing}: The system's administrator schedules round times to accommodate various delays, such as network latency. Balancing security and performance is a consideration in this context.

\textit{Checkpoint Timing}: Checkpoint times are set to align with the average service runtime in cooperation rings. This grants verification teams ample time to submit their votes. Furthermore, the system's resilience isn't reliant on a single verifier's vote due to majority-based decision-making.

\textit{Account Creation Payment}: Registration for the system involves a payment for each account, discouraging frivolous account creation.

\textit{Encouraging Single-Account Usage}: Users are incentivized to operate within the $\LoR$ system using a single account.

Let's examine key attack instances and their corresponding $\LoR$ responses.

\textbf{Double Spending}: This occurs when a coin is used twice. $\LoR$ tackles this by employing mutually exclusive $\RCB$ and $\TCB$ tables. A deliberate delay can be introduced to allow repeated checks, though at the expense of system speed.

\textbf{Long Delay Attack}: An attack involving a task with substantial duration. $\LoR$ limits new service introduction, but collaboration can span extended periods, aligned with Rule~\ref{rule:longtermcoop}. A reliable environment supports cooperation, and any dissatisfaction halts a round's cooperation.

\textbf{Theft of Resources}: Unauthorized resource use without payment is deterred as traders must initially invest in the primary currency, and misuse within a round is short-lived due to termination based on collaboration satisfaction.

\textbf{Honesty Attack}: Discouraging dishonest claims about offering services or payments. Regular checks and penalties are imposed by $\LoR$, and the round-processing time can be adjusted to be very short.

\textbf{Iterative False Generation}: Attack types like Sybil and Peer Flooding, which involve creating numerous pseudonymous identities. $\LoR$ mitigates these by imposing significant costs for account creation and random collaborator selection.

\textbf{Adversarial Centralization of Consensus Power}: Collaborative efforts ensure random verification team selection, reducing the risk of centralization. Policies for cooperation ring verification and punishment for dissenting members deter malicious collaboration attempts.

By aligning its policies with these strategic responses, $\LoR$ provides a robust defense against various attack vectors.


\end{document}